\documentclass[12pt, twoside]{article}
\usepackage{amsmath,amsthm,amssymb}
\usepackage{times}
\usepackage{enumerate}

\def\titlerunning#1{\gdef\titrun{#1}}
\makeatletter
\def\author#1{\gdef\autrun{\def\and{\unskip, }#1}\gdef\@author{#1}}
\def\address#1{{\def\and{\\\hspace*{18pt}}\renewcommand{\thefootnote}{}%
\footnote {#1}}%
\markboth{\autrun}{\titrun}}
\makeatother
\def\email#1{e-mail: #1}
\def\subjclass#1{{\renewcommand{\thefootnote}{}%
\footnote{\emph{Mathematics Subject Classification (2010):} #1}}}
\def\keywords#1{\par\medskip
\noindent\textbf{Keywords.} #1}

\newtheorem{theorem}{Theorem}[section]
\newtheorem{corollary}[theorem]{Corollary}
\newtheorem{lemma}[theorem]{Lemma}
\newtheorem{prop}[theorem]{Proposition}
\newtheorem{rem}[theorem]{\bf Remark}

\newtheorem{example}[theorem]{\bf Example}

\numberwithin{equation}{section}

\newcommand{\en}{{\cal E}}
\newcommand{\dom}{\mathrm{dom}}

\newcommand{\vol}{\mathrm{vol}}

\newcommand{\dist}{\mathrm{dist}}
\newcommand{\capa}{\mathrm{cap}}
\newcommand{\meas}{\mathrm{meas}}
\newcommand{\lags}{\mathrm{\lambda^{G,S}}}
\newcommand{\hags}{{H^{G,S}}}
\newcommand{\ZZ}{\mathbb{Z}}

\newcommand{\RR}{\mathbb{R}}
\newcommand{\PP}{\mathbb{P}}
\newcommand{\EE}{\mathbb{E}}
\newcommand{\NN}{\mathbb{N}}

\begin{document}


\baselineskip=17pt


\titlerunning{Dirichlet Laplacians and Uncertainty Principle}

\title{Lower bounds for Dirichlet Laplacians and uncertainty principles}

\author{Peter Stollmann
\and 
G\"unter Stolz}

\date{}

\maketitle

\address{P.\ Stollmann: Fakult\"{a}t f\"{u}r Mathematik,  Technische
Universit\"{a}t Chemnitz, D-09107 Chemnitz, Germany; \email{P.Stollmann@mathematik.tu-chemnitz.de}
\and
G. Stolz (corresponding author): Department of Mathematics,
University of Alabama at Birmingham, 1402 10th Avenue South, Birmingham
AL 35294-1241, USA; \email{stolz@uab.edu}}



\begin{abstract}
We prove lower bounds for the Dirichlet Laplacian on possibly unbounded
domains in terms of natural geometric conditions. This is used to derive 
uncertainty principles for low energy functions of general elliptic second 
order divergence
form operators with not necessarily continuous main part.

\end{abstract}

\section{Introduction}

Generalized eigensolutions to energies near the bottom of the spectrum of 
infinite volume Laplacians  should be well spread out in configuration space. 
This can be seen as a version of the uncertainty principle: Low (and thus well 
determined) kinetic energy of a quantum particle can not occur simultaneously 
with a sharp concentration of the position of the particle. Mathematically, 
this is usually associated with quantitative forms of unique continuation for 
solutions of second order linear differential equations. Starting with the groundbreaking
work of Carleman, \cite{Carleman} we mention  
\cite{Ag,ABG,Aro,H,JK} for but a small list of references as well as \cite{Koch-Tataru}, which contains
a good overview of the literature up to about 2006. 

While this is a classical topic, it has found renewed interest in recent years 
in the connection with describing the fluctuation boundary regime of 
localization in Anderson-type models with a random potential which only 
partially covers configuration space (also referred to as ``trimmed'' Anderson 
models by some authors, e.g.\ \cite{EK13,RM14}). Eigenvectors or generalized 
eigenvectors of the unperturbed Hamiltonian have to feel the random 
perturbation in order to see a Lifshitz tail regime and lead to an associated 
Wegner estimate. The starting point of this development was the celebrated 
paper \cite{BK}, by Bourgain and Kenig, 
who were the first who could treat the Bernoulli--Anderson model and used 
uncertainty principles in their analysis. For the subsequent development in 
this direction see \cite{BTV,BoKle, GeKle,Klein-13,KT-16,RMV-13,NTTV}
and the references therein. 
Recently, in \cite{DingSmart,LiZhang}, the authors were
able to treat the discrete Bernoulli--Anderson model in 2 and 3 dimensions, respectively, again
with an uncertainty principle as a main ingredient in the proof.

This has provided ample motivation for more thorough studies of the geometric 
properties required for subsets of configuration space to guarantee that these 
subsets carry a ``substantial'' part of the mass of low energy states of the 
Laplacian, both in the continuum setting and for discrete Laplacians on graphs. 
Our goal here is to establish a result in the continuum, 
similar to work in the discrete setting done in \cite{LSS-17} and we refer to 
the literature cited in the latter paper. We mention that from a harmonic 
analysis point of view, our results are close in spirit to Logvinenko--Sereda 
theorems, see \cite{Kov},
with the important difference that we have to restrict ourselves to spectral
projectors with energy intervals close to the ground state energy.

We should stress that the kind of uncertainty principle we aim at is complementary to the classical unique continuation
results referred to above: on the one hand, it does not imply global vanishing of eigensolutions that vanish
on a ball or vanish at some point to infinite order. On the other hand, it can be established in situations, where
such classical unique continuation results are known not to hold: for certain graphs and elliptic divergence operators with
discontinuous main part, which is the focus of the present paper. Let us moreover point out that an important issue is the 
uniformity of estimates with respect to the coefficients and with repect to the underlying domain.

This will allow to prove localization at band edges for new classes of random models, an application we will not work out here (compare e.g.\ \cite{RMV-13} for the use of results on quantitative unique continuation in localization proofs).

Here is the set-up for our main result:

Let $d\ge 3$ and $H^G$ (one half times) the Neumann Laplacian, characterized by 
the 
quadratic form 
\begin{equation} \label{Dirform}
\en{[u]} := \frac{1}{2} \int_G |\nabla u(x)|^2\,dx \quad \mbox{on $W^{1,2}(G)$},
\end{equation}
on an open and convex, not necessarily bounded, domain $G$ in $\RR^d$. The 
reason for including the factor $1/2$ here and in the following is that we will 
study $\en$ through 
its associated Markov process and we want to get the usual Brownian motion 
for $\Omega= \RR^d$. By $P_I(H^G)$ we denote the spectral projection for $H^G$ 
onto an interval $I$.

The inradius of $G$ is
\begin{equation} \label{inradius}
R_G:=\sup\{ r\mid \exists x\in G: B_r(x)\subset G\}\in (0,\infty] .
\end{equation}
Let $R\ge \delta>0$. A closed subset $B\subset G$ is said to be {\it 
$(R,\delta)$-relatively dense} in $G$ with covering radius $R$ and thickness 
$\delta$ provided 
\begin{equation} \label{fatdense}
\forall \,\, x\in G \;\; \exists \,\, y\in B:\;\;B_R(x) \cap B \supset 
B_{\delta}(y). 
\end{equation}
Note that this trivially implies that $\delta \le R_G$.

In this language a set is relatively dense (in the classical sense) if it is 
$(R,0)$-relatively dense for some $R>0$. Typical $(R,\delta)$-relatively dense 
sets are given by fattened relatively dense sets, i.e., their 
$\delta$-neighborhoods.

The main result of our work is the following quantitative unique continuation 
bound for low energy states of $H^G$ and, more generally, elliptic second order 
divergence form operators  of the type $-\nabla \mathbf{a} \nabla$ with 
$\mathbf{a}\in L^{\infty}$, $\mathbf{a} \ge \eta_0$ 
that we introduce now:

Assume that $\mathbf{a}(x), x\in G$ is a symmetric $d\times d$ -- matrix, whose 
coefficients 
are bounded measurable functions of $x$ such that
\begin{equation}
 \label{elliptic}
 \mathbf{a}(x)\ge \eta_0>0. 
\end{equation}
Denote by $H^G_{\mathbf{a}}$ the unique selfadjoint operator defined by the form

\begin{equation} \label{ellform}
\en_{\mathbf{a}}{[u]} := \frac{1}{2} \int_G \left(\mathbf{a}(x)\nabla u(x)\mid 
\nabla u(x)\right)\,dx \quad \mbox{on $W^{1,2}(G)$},
\end{equation}
where we use $\left(\cdot\mid\cdot\right)$ for the inner product in $\RR^d$. 
\begin{theorem} \label{main}
Let $d\ge 3$. Then there exist constants $a,b,C,c>0$, only depending on $d$, 
such that for every open and 
convex $G\subset \RR^d$, any $(R,\delta)$--relatively dense $B$ in $G$, and every elliptic 
$\mathbf{a}$ as in \upshape{(\ref{elliptic})} above,
\begin{equation} \label{eq:main}
\| f 1_B\|^2  \ge \eta_0\kappa \| f\|^2
\end{equation}
for all $f$ in the range of $P_I(H^G_{\mathbf{a}})$, where
\begin{equation} \label{Ikappa}
I=[0,C\eta_0\frac{\delta^{d-2}}{R^d} ]\mbox{  and  
}\kappa=c\left(\frac{\delta}{R}\right)^d\left[ \frac{b}{(R\wedge R_G)^2}
+\left|\log \frac{a\delta^{d-2}}{R^d} \right|\right]^{-2}.
\end{equation}
\end{theorem}

While our method of proof allows estimates only for low energies, the bound in 
(\ref{eq:main}) 
is quite satisfactory. It only differs from the optimal estimate 
$\left(\frac{\delta}{R}\right)^d$ (attained
for constant functions) by a logarithmic correction term and is much better than 
what appears in the literature
so far, see \cite{NTTV} for a comparison.

Maybe more importantly, it is the first uncertainty principle in $d\ge 3$ that 
holds without any
continuity or smoothness assumption on the coefficient matrix $\mathbf{a}$. 
Usual PDE--techniques are known to break down beyond
Lipshitz continuity of the main coefficient, as can be seen from the examples 
in 
\cite{Man,Mil}. 

A nice feature of our method of proof is that we can mainly concentrate on the 
easier case of the Laplacian $H^G$. The uncertainty principle then easily 
extends to
any operator bounded below by a positive multiple of $H^G$ which covers the 
above case
of elliptic second order operators in divergence form. We could as well add 
positive
potentials and consider other boundary conditions, as long as a lower bound is 
available.
For a more complete discussion and possible applications
we refer to Section \ref{sec:uncert} below.

Our proof of Theorem~\ref{main} consists of three parts, covered in the 
remaining three sections of this paper. The same general strategy has been used 
in \cite{LSS-17} to prove corresponding results for Laplacians on graphs. The 
continuum setting considered here leads to some additional complications.

The overall idea, presented in the conclusion of the proof of our main theorem 
in Section~\ref{sec:uncert}, is to reduce the uncertainty principle 
\eqref{eq:main} to showing that the bottom of the spectrum of $H^G + \beta 
1_B$ rises above the energy interval $I$ in the large coupling limit $\beta 
\to \infty$ (where $1_B$ is the characteristic function of $B$). This 
approach to uncertainty principles was introduced in \cite{BLS-10}. It provides 
explicit lower bounds on $\kappa$ which will yield \eqref{Ikappa}. 

So we have to understand the large coupling limit of $H^G + \beta 1_B$, 
which we start in Section~\ref{LowerBounds} by studying the case of infinite 
coupling. This means we will find a lower bound for $H^{G,S}$, the Laplacian 
on $\Omega := G \setminus S$ with Neumann condition on the boundary of $G$ and 
an additional Dirichlet condition on the boundary of a set $S$ (whose relation 
to $B$ will be explained below). It is here where we encounter one of the main 
differences between the discrete and continuous case: Points in $\RR^d$, for 
$d\ge 2$, are not massive in the sense of $1$-capacities. We further illustrate 
this in Appendix~\ref{capacities} by providing a simple  (and certainly not 
new) example of a set with finite inradius whose Dirichlet Laplacian has 
spectrum $[0,\infty)$. The key insight in this part of our proof is that we can 
quantify how lower bounds of Dirichlet Laplacians with $\delta$-fat and 
relatively dense complement depend on $\delta$. The crucial geometric quantity 
we identify
in Theorem \ref{mainsec2} can be interpreted as the capacity per unit volume of
the set $S$ of obstacles, reminiscent of the ``crushed ice problem'', see 
Section~\ref{LowerBounds}.

As a last part of the strategy we need to be able to relate the the lower 
bounds for finite and infinite coupling, respectively. Here it is crucial for 
our proof that the set $S$ is chosen as a slightly smaller (``semi-fat'') 
version of $B$. The space created between the boundaries of $B$ and $S$ will 
allow to compare the spectral minima of $H^{G,S}$ and $H^G + \beta 1_B$ via 
a norm bound on the difference of the corresponding heat semigroups. The latter 
bound will be proven via the Feynman-Kac formula in Section~\ref{sec:hitrun}. 
In particular, this will use a 'hit and run' Lemma which bounds the probability 
that a Brownian path can hit the center of a fat set and then leave the set (by 
crossing the space between $B$ and $S$) within a short time.

In addition to our main result, some of the auxiliary results obtained in 
Sections~\ref{LowerBounds} and \ref{sec:hitrun} should be of independent 
interest. The lower bounds on Dirichlet Laplacians of sets with 
$(R,\rho)$-relatively dense complement shown in Theorem~\ref{mainsec2} improve 
on a classical result in \cite{Davies} in their dependence on the ratio 
$\rho/R$ (and allow for an additional Neumann part of the boundary), see the 
comments at the end of Section~\ref{LowerBounds}. Also, while the `hit and run' 
Lemma~\ref{lem:hitrun} has been used in spectral theory before (e.g.\ 
\cite{McStolzman}), we feel that this tool deserves additional advertising. 
Moreover, as we point our here, it also holds for reflected Brownian motion, 
i.e., in the study of the heat semigroup of Neumann Laplacians.

\section{Lower bounds for the Dirichlet Laplacian on unbounded domains with 
uniform relatively dense complement} \label{LowerBounds}

The first ingredient into our strategy of proof are quantitative lower bounds 
for Dirichlet Laplacians $-\Delta_\Omega$ on sets $\Omega$ with ``fat'' 
relatively dense complement in the sense of \eqref{fatdense}. More generally, 
we consider a set-up where this is done relatively to a convex open subset $G$ 
of $\RR^d$, on whose boundary we will place a Neumann condition. The assumption 
on $G$ could be weakened in some ways, but we make it for clarity and 
because it provides a convenient class of sets which have all the properties 
required for our proofs.

In particular, we will use that convex sets are star shaped and that 
intersections of convex sets are convex. Also, convex sets satisfy the segment 
property and thus, by Theorem~3.22 in \cite{adams}, we have the first claim in
\begin{eqnarray} \label{density}
& \{u\in C^1(G) \cap C_c(\overline{G}):\|u\|_{1,2}<\infty\}  \;\;\mbox{is dense 
in} & \notag \\ & (W^{1,2}(G), \|\cdot\|_{1,2}) \;\;\mbox{and in}\;\; 
(C_c(\overline{G}), \|\cdot\|_{\infty}). &
\end{eqnarray}
Here $\overline{G}$ denotes the closure of $G$ and 
\begin{equation} 
\|u\|_{1,2} = \left( \int_G (|\nabla u(x)|^2 + |u(x)|^2)\,dx\right)^{1/2}
\end{equation} the Sobolev norm. The second 
claim in \eqref{density} can be seen from the Stone-Weierstrass Theorem: To 
$f\in 
C_c(\overline{G})$ let $K:= \mbox{supp}\,f$ and choose an open ball $U$ and a 
closed ball $B$ in $\RR^d$ such that $K \subset B \subset U$. Stone-Weierstrass 
shows that $\{\varphi|_B: \varphi \in C_c^{\infty}(U)\}$ is dense in $C(B)$, 
i.e., there exist $\varphi_n \in C_c^{\infty}(U)$ such that $\sup_{x\in B} 
|\varphi_n(x)-f(x)| \to 0$. Finally, choose $\chi \in C_c^{\infty}(U)$ such 
that $\chi|_K=1$. Then $\chi \varphi_n \in C^1(G) \cap C_c(\overline{G})$ with 
$\sup_{x\in \overline{G}}|(\chi \varphi_n)(x)-f(x)| \to 0$.

Note that $C_c(\overline{G})$-functions are not supposed to vanish at the 
boundary of $G$. Therefore we get that the form \eqref{Dirform} 
can be regarded as a {\it regular Dirchlet form} on the locally compact space 
$\overline{G}$, see \cite{Fukushima} for basics on Dirichlet forms and 
potential theory. In particular, 
there is a process 
associated with $\en$, via reflected Brownian motion, a fact that will be of 
primary importance in the 
sequel.

Let $H=H^G$ (mostly, we omit the superscript) be the associated 
Laplacian, which is $-\frac{1}{2}\Delta$ in $L^2(G)$ with Neumann boundary 
conditions. The Dirichlet Laplacians referred to in the title are given by an 
additional Dirichlet boundary condition on a closed set $S$, which is defined 
via forms again through $\Omega := \overline{G} \setminus S$ 
and
\begin{equation} \label{EOmega}
\en^{G,S} = \en \:\: \mbox{on $\dom(\en^{G,S}) = \overline{ 
\{u\in C^1(G) \cap C_c(\Omega) :\|u\|_{1,2}<\infty\} }^{W^{1,2}}$}.
\end{equation}
As will be discussed in Section~\ref{sec:hitrun} below, this form is associated 
with a process that is related to the one of $H$ by killing paths 
once they hit the set $S$. Note that $\en^{G,S}$ is closed and densely 
defined in $L^2(\Omega)$ and denote the associated mixed Neumann-Dirichlet 
Laplacian on $L^2(\Omega)$ by $H^{G,S}$.

The main result of this section is a lower bound for 
\begin{equation}
\lambda^{G,S} :=  \inf \sigma(H^{G,S}) 
\end{equation}
whenever $\Omega:= \overline{G} \setminus S$ for an $(R,\rho)$-relatively dense 
closed subset $S$ of $G$. Note that, by \eqref{EOmega} and the variational 
principle,
\begin{equation} \label{varprin}
\lambda^{G,S} = \inf \left\{ \frac{1}{2} \int_\Omega |\nabla u(x)|^2\,dx: u \in 
C^1(G) \cap C_c(\Omega),\;\int_\Omega |u(x)|^2\,dx=1 \right\}.
\end{equation}

We start with a finite volume estimate: Here we let $G$ be open and convex and 
assume in addition that 
\begin{equation} \label{rhoR}
B_{\rho}(0) \subset G \subset B_R(0)\quad  \mbox{for 
$0<\rho<R<\infty$.}
\end{equation}
Denote $\mathbb{S}^{d-1} = \{ u\in \RR^d \,|\, |u|=1\}$ and by $du$ the surface 
measure on $\mathbb{S}^{d-1}$ induced by Lebesgue measure on $\RR^d$. Let 
$\omega_d$ be the volume of the $d$-dimensional unit ball. Let $R: 
\mathbb{S}^{d-1} \to (\rho,R]$ be the ``radius function'' of $G$, i.e., $R_u := 
\sup \{t\,|\, tu\in G\}$ for $u \in 
\mathbb{S}^{d-1}$. Note that $R$ is lower semicontinuous and hence measurable.

\begin{prop} \label{localbound}
Let $d\ge 3$, $G$ open and convex satisfying \eqref{rhoR}, $S:= B_{\rho}(0)$ 
and $\Omega := \overline{G} 
\setminus S$. Then for $H^{G,S}$ defined as above, we have
\begin{equation} \label{lowbound}
d(d-2) \frac{\rho^{d-2}}{R^d} \le \lambda^{G,S}.
\end{equation}
If, furthermore, $B_{2\rho}(0) \subset G$, then
\begin{equation} \label{upbound}
\lambda^{G,S} \le 2^d  \omega_d \,\frac{\rho^{d-2}}{\vol{(G)} - \omega_d 
(2\rho)^d}.
\end{equation}
\end{prop}

\begin{proof}
For the lower bound, by \eqref{varprin}, it suffices to consider $f\in C^1(G)$, 
$f=0$ on 
$B_{\rho}(0)$, and prove an estimate for $\|f\|_2^2$ in terms of $\| \nabla 
f\|_2^2$. So let $u\in \mathbb{S}^{d-1}$, $r\in [\rho, R_u)$. We have $f(ru) = 
\int_{\rho}^r \partial_u f(tu)\,dt$ and thus
\begin{eqnarray}
|f(ru)|^2 & \le & \int_{\rho}^r |\partial_u f(tu)|^2 t^{d-1}\,dt \cdot 
\int_{\rho}^r t^{1-d}\,dt \notag \\
& \le & \int_{\rho}^r |\partial_u f(tu)|^2 t^{d-1}\,dt \cdot \frac{1}{d-2} 
\rho^{2-d}.
\end{eqnarray}
Integrating with respect to surfaces we get
\begin{eqnarray} \label{keycalc}
\|f\|^2 & = & \int_G |f(x)|^2\,dx = \int_{\mathbb{S}^{d-1}} \int_0^{R_u} 
|f(ru)|^2\,dr\,du \notag \\
& \le & \int_{\mathbb{S}^{d-1}} \int_{\rho}^{R_u} r^{d-1} \int_{\rho}^r |\nabla 
f(tu)|^2 t^{d-1}\,dt \frac{1}{d-2} \frac{1}{\rho^{d-2}} \,dr\,du \notag \\
& \le & \frac{1}{(d-2)\rho^{d-2}} \int_{\rho}^{R} r^{d-1}\,dr 
\int_{\mathbb{S}^{d-1}} \int_{\rho}^{R_u} |\nabla f(tu)|^2 t^{d-1}\,dt\,du 
\notag \\
& \le & \frac{1}{d(d-2)} \frac{R^d}{\rho^{d-2}} \|\nabla f\|_2^2,
\end{eqnarray}
which gives the asserted lower bound.

The upper bound can be shown by a test function of the following form: $f(x)= 
\varphi(|x|)$ where $\varphi(s)=0$ for $0\le s \le \rho$, $\varphi(s)= 
(s-\rho)/\rho$ for $\rho < s \le 2\rho$ and $\varphi(s)=1$ for $s>2\rho$. It 
follows that $|\nabla f(x)| = \rho^{-1} \cdot 1_{B_{2\rho}(0) \setminus 
B_{\rho}(0)}$ and therefore
\begin{equation}
\|\nabla f\|_2^2 \le \rho^{-2} \vol{(B_{2\rho}(0))} = 2^{d} \omega_d \rho^{d-2}.
\end{equation}
The assertion now follows from $\|f\|_2^2 \ge \vol{(G)} - \vol{(B_{2\rho}(0)}$.
\end{proof}

\begin{rem}
(i) We think of $\rho$ as small compared to $R$ and $G$ a set of almost the 
size of 
$B_R$. In such a case the upper and lower bounds in the preceding proposition 
match up to constants and are both of the form
\begin{equation} \label{twosidebound}
\frac{\rho^{d-2}}{R^d}.
\end{equation}

(ii) One can modify the above calculations to get bounds for $d=2$, but due to 
the appearing logarithmic terms we do not easily see a two-sided bound 
comparable to \eqref{twosidebound} in this case. This is the main reason why, 
here and in the following, we limit our discussion to $d\ge 3$.
\end{rem}

As a first special case of the main result of this section 
(Theorem~\ref{mainsec2} below), we go on to apply the above local result to a 
standard geometric situation considered in 
recent unique continuation results, e.g.\ \cite{BTV,RMV-13}. For obvious 
reasons it is called a ``ball pool'' by some experts in the field. The lower 
bound we present is a 
first step towards a quantitative unique continuation estimate that is very 
explicit as far as constants are concerned. 

Consider $\rho>0$ and $\ell>0$ with $\rho<\ell/2$ and a sequence of balls 
$B_\rho(y_k) \subset k+(0,\ell)^d$, $k\in (\ell\ZZ)^d$. Let $\Gamma \subset 
(\rho\ZZ)^d$ be an arbitrary subset of lattice points and
\begin{equation} \label{ballpool}
S:= \bigcup_{k\in \Gamma} B_{\rho}(y_k).
\end{equation}
S is contained in the interior 
\begin{equation} \label{ballpoolG}
G = \left( \bigcup_{k\in \Gamma} (k+[0,\ell]^d)\right)^{\circ}
\end{equation}
of the corresponding union of closed cubes.
Clearly, this gives an example of a set $S$ which is $(R,\rho)$-relatively 
dense in $G$ for 
$R=\sqrt{d} \ell$. 

\begin{corollary} \label{cor:ballpool}
Let $S$ and $G$ be given by \eqref{ballpool} and \eqref{ballpoolG}. 
Consider $H^{G,S}$ as defined above with $\Omega := \overline{G} 
\setminus S$. Then
\begin{equation} \label{speclowbound}
\lambda^{G,S} \ge (d-2) \frac{(\rho/\sqrt{d})^{d-2}}{\ell^d}
\end{equation}
\end{corollary}

\begin{proof}
By \eqref{varprin}, it suffices to bound $\|f\|_2^2$ in terms of $\|\nabla 
f\|_2^2$ for any $f\in 
C^1(G) \cap C_c(\Omega)$. This follows easily from \eqref{keycalc}, applied to 
each of the sets $\Omega_k:= (k+(0,\ell)^d) \setminus B_{\rho}(y_k)$, since 
$B_{\rho}(y_k) \subset k+(0,\ell)^d \subset B_{\ell\sqrt{d}}(y_k)$:
\begin{eqnarray}
\|f\|_2^2 & = & \sum_{k\in \Gamma} \|f 1_{k+(0,\ell)^d}\|^2 \le 
\frac{1}{d(d-2)} 
\frac{(\ell\sqrt{d})^d}{\rho^{d-2}} \sum_{k\in \Gamma} \|\nabla f \cdot 
1_{k+(0,\ell)^d}\|_2^2 \notag \\ & = & \frac{1}{d-2} 
\frac{\ell^d}{(\rho/\sqrt{d})^{d-2}}
\|\nabla f\|_2^2.
\end{eqnarray}
\end{proof}

The main result of this section, Theorem~\ref{mainsec2} below, extends this to 
general $(R,\rho)$-relatively 
dense subsets $S$ of open and convex sets $G$, without requiring the specific 
geometry used in Corollary~\ref{cor:ballpool}. We start
with a preliminary geometrical result that will help in the sequel. 
\begin{prop}\label{skeleton}
 Let $d\ge 3$, $G\subset \RR^d$ open and convex, and
$S\subset G$ be $(R,\rho)$-relatively dense in $G$. Then there is a 
$\Sigma\subset S$ with
the following properties:
\begin{itemize}
\item[(a)] $B_\rho(\Sigma):=\bigcup_{p\in \Sigma} B_{\rho}(p)$ is 
$(3R,\rho)$-relatively dense in $G$ 
and $B_\rho(\Sigma)\subset S$.
\item[(b)] $\bigcup_{p\in \Sigma} B_{3R}(p) \supset \overline{G}$,
\item[(c)] If $p \in \Sigma$ and $\Sigma\setminus\{ p\}\not=\emptyset$, then 
\begin{equation} 
R\le \dist(p,\Sigma\setminus\{ p\}) \le 6R,
\end{equation}
in particular, $\Sigma$ is uniformly discrete and $B_\rho(\Sigma\setminus\{ 
p\})$ is $(6R,\rho)$-relatively dense 
in $G$. 
\end{itemize}
\end{prop}
We call such a set $\Sigma$ a \emph{skeleton} of $S$.
\begin{proof}
$(R,\rho)$-relative denseness of $S$ ensures that we find a subset $D\subset S$ 
such 
that
\begin{equation} \label{GcontB}
\bigcup_{p\in D} B_R(p) \supset G
\end{equation}
and $B_{\rho}(p) \subset S$ for any $p\in D$. We may pick a subset $\tilde{D}
\subset D$ such that
\begin{equation} \label{pqgeR}
p,q\in \tilde{D} \Longrightarrow |p-q|\ge R,
\end{equation}
i.e., so that $\tilde{D}$ is uniformly discrete, which is nothing but the lower 
bound appearing in (c). By Zorn's lemma, there 
exists 
a maximal subset $\Sigma \subset D$ 
with this property. Then $\Sigma$ satisfies (a), (b) and (c): 

By construction it satisfies $B_\rho(\Sigma) \subset S$ and the lower bound in 
(c), i.e., uniform discreteness.

To show (b), assume that there is $x\in G$ such 
that $\dist{(x,\Sigma)} \ge 3R$. By \eqref{GcontB} the ball 
$B_R(x)$ contains at least one $p_0 \in D$. The triangle inequality yields that 
$\{p_0\} \cup \Sigma$ still satisfies \eqref{pqgeR}, contradicting the assumed 
maximality of $\Sigma$. This shows $\bigcup_{p\in \Sigma} B_{3R}(p) 
\supset G$. The union on the left side is closed (by the uniform discreteness), 
so that (b) follows. This readily gives that $B_\rho(\Sigma)$ is 
$(3R,\rho)$-relatively dense in $G$, completing the verification of (a). 
 
 We are left to prove the upper bound in (c) under the assumption $p \in 
\Sigma$ 
and $\Sigma\setminus\{ p\}\not=\emptyset$.
 So let $R':= \dist(p,\Sigma\setminus\{ p\})=|p-q|$ for $q\in \Sigma\setminus\{ 
p\}$. By 
 uniform discreteness of $\Sigma$, we can find such a $q$. The midpoint $s$ of 
the line segment $[p,q]$ belongs to $G$ by convexity
 and so there is an $s'\in \Sigma$ such that $|s-s'|\le 3R$. The minimality of 
$|p-q|$ gives that $|s-q|\le 3R$ as well, settling
 that $|p-q|=2|s-q|\le 6R$.
 \end{proof}

\begin{theorem} \label{mainsec2}
Let $d\ge 3$, $G\subset \RR^d$ open and convex, and
$S\subset G$ be $(R,\rho)$-relatively dense in $G$.
Then, for $\Omega := \overline{G} \setminus S$, we have
\begin{equation} \label{genlowbound}
\lambda^{G,S} \ge \frac{d(d-2)}{3^d} \frac{\rho^{d-2}}{R^d}.
\end{equation}
\end{theorem}

\begin{proof}
Firstable, notice that by monotonicity it suffices to prove a bound for any 
subset
$\tilde{S} \subset S$.

We pick $\tilde{S}=B_\rho(\Sigma)$, where $\Sigma$ is a skeleton of $S$, the 
existence of which is granted
by Proposition \ref{skeleton} above. Define the corresponding Vorono\"i 
decomposition of $\overline{G}$ by
\begin{equation}
G_p := \{x\in \overline{G} \;|\; |x-p| \le |x-q| \:\mbox{for all $q\in 
\Sigma$}\}, \quad p\in \Sigma.
\end{equation}
By construction we see that
\begin{itemize}
\item[(i)] $\bigcup_{p\in \Sigma} G_p = \overline{G}$,
\item[(ii)] $\stackrel{\circ}{G_p} \cap \stackrel{\circ}{G_q} = \emptyset$ for 
$p, q \in \Sigma$, $p \not= q$,
\item[(iii)] $B_{\rho}(p) \subset G_p \subset B_{3R}(p)$
\end{itemize}
and $G_p$ is the intersection of $\overline{G}$ with a finite number of 
half-spaces. In particular, all the sets $G_p$ as well as their interiors are 
convex.

To prove the assertion of the Theorem, it suffices to bound $\|f\|_2^2$ 
appropriately in terms of $\|\nabla f\|_2^2$ for given $f\in 
C^1(G) \cap C_c(\Omega)$. Note that 
(iii) above allows us to apply Proposition~\ref{localbound} to $G_p$ with $R$ 
replaced by $3R$. Therefore we get, also using (i) and (ii),
\begin{eqnarray}
\|f\|_2^2 & = & \sum_{p\in \Sigma} \|f 1_{G_p}\|_2^2 \le \frac{1}{d(d-2)} 
\frac{(3R)^d}{\rho^{d-2}} \sum_{p\in \Sigma} \| (\nabla f) 1_{G_p}\|_2^2 
\notag \\
& = & \frac{3^d}{d(d-2)} \frac{R^d}{\rho^{d-2}} \|\nabla f\|_2^2.
\end{eqnarray}
\end{proof}

We remark that wanting to work with a Vorono\"i decomposition required to 
choose a uniformly discrete skeleton $\Sigma$ of $D$ in the above proof. This 
is 
the reason why the constants in \eqref{genlowbound} and the special case 
\eqref{speclowbound}, where the Vorono\"i cells are given a priori, differ by a 
factor $3^d$.

In case $G= \RR^d$, we could employ Theorem~1.5.3 from \cite{Davies} which 
gives a lower bound on $H^{\RR^d,S}$, the Dirichlet Laplacian on $\Omega = 
\RR^d \setminus S$, in terms of
\begin{equation}
d_u(x) := \min \left\{ |t| \,\big|\, x+tu \in S \right\}, \quad u\in 
\mathbb{S}^{d-1}, 
\end{equation}
\begin{equation}
\frac{1}{m(x)^2} := \int_{\mathbb{S}^{d-1}} \frac{du}{d_u(x)^2}.
\end{equation}
More precisely,
\begin{equation}
H^{G,S} \ge \frac{d}{8m^2}
\end{equation}
in the sense of quadratic forms. In the case at hand and in the regime $0<\rho 
<< R$ we could bound $d_u(x)$ by 
$R$ on a set of unit vectors of size $\rho^{d-1}/R^{d-1}$, so that we would get 
a lower bound on $\lambda_\Omega$ of the form
\begin{equation}
\mbox{const} \frac{\rho^{d-1}}{R^{d+1}},
\end{equation}
which is worse (by a factor $\rho/R$) than what we have proven above. More 
importantly, it is not clear how to adapt Davies' method of proof to the case 
of the Neumann Laplacian on subdomains.

It is well known that the capacity of a ball of radius $r$ in $\RR^d$ behaves 
like
$r^{d-2}$ for $d\ge 3$ and small $r\ge 0$, see the discussion in Appendix 
\ref{capacities} 
below. For well--spaced $S$ that means that the crucial geometric property of 
$S$ that determines 
the lower bound in \eqref{genlowbound} can be regarded as the capacity per unit 
volume.

This is well in accordance with the results for the ``crushed ice problem'' in 
the celebrated
article \cite{RT} by Rauch and Taylor.

We will now discuss some consequences of Theorem \ref{mainsec2} for related 
situations that
shed some light on ``singular homogenization'' in the following sense.

Fix $G\subset \RR^d$ for $d\ge 3$ and consider a sequence $S_n$ of sets that 
are $(R_n,\rho_n)$--relatively
dense. We think of each $S_n$ as a union of $\rho_n$--balls with radius 
$\rho_n\to 0$ as $n\to\infty$.
If we increase the number of balls so that
\begin{equation}\label{nonfading}
\inf_{n\in\NN}\frac{\rho_n^{d-2}}{R_n^d}>0,
\end{equation}
the presence of the tiny obstacles will be felt in the limit, since there is a 
uniform lower bound
for the operators $H^{G,S_n}$ by \eqref{genlowbound} above. 

If
\begin{equation}\label{solid}
\frac{\rho_n^{d-2}}{R_n^d}\to\infty\mbox{  for  }n\to\infty,
\end{equation}
the operators $H^{G,S_n}$ ``diverge to $\infty$'' in the sense that
$$
\| \left(H^{G,S_n}+1\right)^{-1}\|\to 0\mbox{  for  }n\to\infty,
$$
again by \eqref{genlowbound} above. To relate this behaviour to the set--up in 
\cite{RT}, let us specialize
to the case where $G$ is bounded and $S_n$ consists of $n$ balls of radius 
$\rho_n$ (called $r_n$ in the
above paper. There it is shown that for $n\rho_n^{d-2}\to 0$, the effect of the 
small holes
vanishes in the limit, the obstacles are fading. This is a consequence of the 
fact that the capacity
of $S_n$ tends to $0$ in this case. Actually, using Theorem 1 from 
\cite{StollmannJFA}, it follows that
the semigroup of $H^{G,S_n}$ converges to the semigroup of $H^G$ in Hilbert 
Schmidt norm, which gives a quite strong
convergence result. A volume counting argument shows that
$$
n\sim R_n^{-d},
$$
so that we recover the different phases identified in \cite{RT}, who study the 
limit of the operators while
we restrict to the analysis of lower bounds. However, the estimates in 
\eqref{nonfading} and \eqref{solid} give
information for fixed configurations, in contrast to what is found in \cite{RT}.

\section{A norm estimate for the heat semigroup at large coupling} 
\label{sec:hitrun}

In comparison with the discrete case, \cite{LSS-17}, this is probably the most 
tricky part
of the present analysis.

We fix an open and convex set $G$ and a closed $(R,\rho)$-relatively dense 
subset $S$ of $G$, and
\begin{equation}
B := B_{\rho}(S).
\end{equation}
To get a lower bound for eigenfunctions of $H=H^G$ we will use a lower bound on
\begin{equation}
\lambda_{\beta} := \inf \sigma(H_{\beta}),
\end{equation}
where $H_{\beta} := H+ \beta 1_B$. To this end, we will introduce an additional 
Dirichlet boundary condition on $S$ and compare, in this section, 
$e^{-H_{\beta}}$ and $e^{-H^{G,S}_{\beta}}$ in the operator norm. Here $\Omega := \overline{G} \setminus S$ and
\begin{equation}
H^{G,S}_{\beta} = H^{G,S} + \beta 1_{B\setminus S}
\end{equation} 
on $L^2(\Omega)$ and, as usual, $e^{-H^{G,S}_{\beta}}$ is interpreted as an operator on 
$L^2(G)$ by setting it $0$ on $L^2(S)$. 

The main idea is that this additional Dirichlet boundary condition at $S$ does 
not 
matter too much for large $\beta$, since the potential barrier given by $\beta 
1_{B\setminus S}$ is almost impenetrable from within $\Omega$. To formalize and 
quantify this heuristic we use 
the probabilistic representation of the semigroup, the Feynman-Kac formula, 
that gives how the potential and the Dirichlet boundary condition enter the 
probabilistic formulae and, most importantly, the ``hit and run'' Lemma that 
shows that, with an overwhelming probability, each Brownian path that hits $S$ 
stays around at least for some time in the $\rho$-neighborhood $B$ of $S$.

This additional twist is necessary, since there are no quantitative results that
allow to control the convergence of $\lambda_\beta$ as $\beta\to\infty$ 
directly. We refer
to \cite{BAB,BD} and the results cited there for partial results.

We first record some basic facts. Since, by assumption, $H$ corresponds to a 
regular Dirichlet form, 
by \cite{Fukushima}, Thm 6.2.1, p.~184  there is a process $(\mathbf{\Omega}, 
(\PP_x)_{x\in \overline{G}}, (X_t)_{t\ge 0}, 
({\mathcal F}_t)_{t\ge 0})$ which is associated with $H$ in the sense that for 
any $t\ge 0$ and $f\in L^p(G)$ ($1\le p \le \infty$):
\begin{equation} \label{BM}
\EE_x(f \circ X_t) = e^{-tH} f(x)
\end{equation}
almost everywhere. Here $\EE_x$ is expectation with respect to $\PP_x$. 

By \cite{Fukushima}, p.~89f  we know that this process has the strong Markov 
property and, since the form is strongly local, the paths are continuous, see 
\cite{Fukushima}, Thm 6.2.2, p. 184. In 
the case at hand, $(X_t)_{t\ge 0}$ is {\it reflected Brownian motion} RBM, 
which coincides with usual Brownian motion on $\RR^d$, denoted $(W_t)_{t\ge 0}$ 
as long as particles do not hit the boundary of $G$. The exact meaning of 
this will be elaborated in our arguments below.

From the general theory we infer the Feynman-Kac formula \cite{Fukushima}
\begin{equation}
e^{-tH_{\beta}} f(x) = \EE_x \left[ f\circ X_t \cdot \exp\left(-\int_0^t \beta 
1_B \circ X_s\,dx\right)\right]
\end{equation}
and denote the appearing {\it occupation time} for $t=1$ by
\begin{equation}
T = T(\omega) = \int_0^1 1_B \circ X_s\,dx = \meas \{s\in [0,1]\,|\, X_s \in 
B\}.
\end{equation}
Moreover, we denote the {\it first hitting time} of $S$ by
\begin{equation}
\sigma := \sigma_S := \inf \{s\ge 0 \,|\, X_s \in S\}
\end{equation}
and infer from \cite{Fukushima} that the additional Dirchlet boundary condition 
kills the Brownian motion, i.e., 
\begin{equation}
e^{-tH^{G,S}} f(x) = \EE_x [ f\circ X_t \cdot 1_{\sigma >t}]
\end{equation}
as well as
\begin{equation} \label{KilledFC}
e^{-tH^{G,S}_{\beta}} = \EE_x \left[ f\circ X_t \cdot \exp \left( - \int_0^t 
\beta 1_B \circ X_s\,ds \right) 1_{\sigma>t} \right].
\end{equation} 
We specialize to $t=1$, where the r.h.s.\ of \eqref{KilledFC} becomes $\EE_x [ 
f\circ X_1 \cdot e^{-\beta T} 1_{\sigma>1}]$. 

\begin{lemma}['Hit and Run'--Lemma] \label{lem:hitrun}
 In the situation above, for $x\in G$,
 \begin{eqnarray} \label{hitrun}
 \PP_x\{ \sigma\le 1,T\le \alpha\}\le 2^{\frac{d}{2}+2}\exp\left( 
-\frac{\rho^2}{16\alpha}\right).
 \end{eqnarray}
\end{lemma}

Let us mention the very convincing intuitive meaning 
of \eqref{hitrun}, at least on a qualitative level: A 
Brownian path belonging to the event in question has to do a full crossing of a 
wall of thickness $\rho$ in time 
at most $\alpha$, i.e., ``hit'' $S$ and then quickly ``run'' away from it 
again. Clearly, the probability for this to happen should be quite small if 
$\rho$ is large and 
$\alpha$ small.

In the case of $G=\RR^d$ the 'hit and run'--lemma was already used for 
spectral theoretic purposes in \cite{McStolzman},
see Lemma~3 in the latter article (see also \cite{Stolzman} for related techniques). Let us briefly explain why reflected 
Brownian motion agrees with the usual
one up to the hitting time of the boundary. For bounded regions, much more 
precise statements are known, see \cite{Chen} and \cite{BC}, where
a calculation quite like the one we use below is presented. Since we allow 
unbounded regions however, these references do not 
settle the case, although it is quite obvious that boundedness should not 
matter. Our argument goes as follows: the process $(X_t)_{t\ge 0}$ in question 
is,
as we saw above, associated with the regular Dirichlet form of $H=H^G$; adding a 
killing or Dirichlet b.c.\ at $\partial G$ results in the same form that one 
obtains
when adding a Dirichlet condition on $G^c$ for the usual Laplacian on $\RR^d$, 
for 
which we get usual Brownian motion $(W_t)_{t\ge 0}$, killed at $\partial G$. 
Since processes 
are essentially uniquely determined by the form, see Theorem 4.2.8 in 
\cite{Fukushima}, 
this means that $(X_t)_{t\ge 0}$ and $(W_t)_{t\ge 0}$ agree up to the time when
they hit $\partial G$.
\begin{proof} 
We introduce the following auxiliary set and stopping time:
\begin{equation}
B' := B_{\rho/2}(S) \subset B,
\end{equation}
and
\begin{equation}
\tau := \inf \{s>0\,| \,X_s \in B'\},
\end{equation}
as well as the event
\begin{equation}
E:= \{ \omega \in \mathbf{\Omega}\,| \,X_0(\omega) \in B' \; \mbox{and} \; 
|X_s(\omega)-X_0(\omega)| \ge \rho/2 \; \mbox{for some} \; s\le \alpha\}.
\end{equation}
Since $B_{\rho/2}(y) \subset B$ for $y\in B'$, $X_s$ agrees with classical 
Brownian motion up to the exit time $\tau_{\rho/2}^W$ for the Wiener process,
\begin{equation}
\PP_x(E) = \PP_0[ \tau^W_{\rho/2} \le \alpha].
\end{equation}
By the reflection principle,
\begin{equation}
\PP_0[ \tau^W_{\rho/2} \le \alpha] \le 2 \PP_0 [ |W_\alpha| \ge \rho/2\}.
\end{equation}
From the explicit formula for the latter we get
\begin{eqnarray}
\PP_0[|W_\alpha| \ge \rho/2] & = & (2\pi\alpha)^{-d/2} \int_{|y|\ge \rho/2} 
\exp \left( -\frac{|y|^2}{2\alpha} \right)\,dy \notag \\
& \le & (2\pi\alpha)^{-d/2} \exp \left( -\frac{\rho^2}{16\alpha} \right) 
\int_{|y|\ge \rho/2} \exp \left( -\frac{|y|^2}{4\alpha} \right)\, dy \notag \\
& \le & 2^{d/2} \exp \left( -\frac{\rho^2}{16\alpha} \right) 
(4\pi\alpha)^{-d/2} \int_{\RR^d} \exp \left( - \frac{|y|^2}{4\alpha} 
\right)\,dy 
\notag \\
& = & 2^{d/2} \exp \left( -\frac{\rho^2}{16\alpha} \right).
\end{eqnarray}
We conclude that 
\begin{equation} \label{PxEbound}
\PP_x(E) \le 2^{\frac{d}{2}+1} \exp \left( - \frac{\rho^2}{16\alpha} \right).
\end{equation}

We go on to estimate the probability in question by
\begin{equation} \label{split}
\PP_x\{ \sigma \le 1, T \le \alpha\} \le \PP_x(\Omega_1) + \PP_x(\Omega_2)
\end{equation}
for the events $\Omega_1 := \{ \sigma \le 1, T\le \alpha, \tau \le 1-\alpha\}$ 
and $\Omega_2 := \{ \sigma \le 1, T \le \alpha, \tau>1-\alpha\}$.

First consider $\Omega_1$ and $x\not\in B'$. In this case, as $X_0(\omega)=x$ 
for $\PP_x$-a.e.\ $\omega \in \Omega_1$, we know by continuity of sample paths 
that $\tau(\omega) \le \sigma(\omega)$ and $X_{\tau(\omega)}(\omega) \in 
\partial B'$. From $T\le \alpha$ we conclude that $\omega$ must leave $B$ 
before $\tau+\alpha$ ($\le 1$). In particular, $\omega$ must leave
\begin{equation}
B_{\rho/2} (X_{\tau(\omega)}(\omega)) \subset B
\end{equation}
and, therefore,
\begin{equation}
(X_{\tau+s}(\omega))_{s\ge 0} \in E.
\end{equation}
Denoting conditional expectation (in $L^{\infty}(\Omega)$) by 
$\EE_{\bullet}$, this can be put together as
\begin{eqnarray}
\PP_x(\Omega_1) & = & \EE_x (\EE_{\bullet} (\Omega_1\,|\, {\mathcal F}_\tau)) 
\notag \\
& \le & \EE_x ( \EE_{\bullet} ( (X_{\tau+s})_{s\ge 0} \in E \,|\, {\mathcal 
F}_\tau )) 
\notag \\
& = & \EE_x ( \PP_{X_{\tau}(\omega)} (  E ))
\end{eqnarray}
by the strong Markov property.  Finally, by \eqref{PxEbound},
\begin{equation} \label{Omega1bound}
\PP_x(\Omega_1) \le 2^{\frac{d}{2}+1} \exp \left( - \frac{\rho^2}{16\alpha} 
\right).
\end{equation}

For $x\in B'$ it is clear that $\tau(\omega)=0$ for $\PP_x$-a.e.\ $\omega \in 
\Omega_1$ and, by the reasoning above, \eqref{Omega1bound} holds in this case 
as well.

Concerning the second term in \eqref{split}, it is clear that 
$\PP_x(\Omega_2)=0$ for $x\in B'$, so we can stick to the case $x\not\in B'$. 
For $\PP_x$-a.e.\ $\omega \in \Omega_2$ we know that $\tau \le \sigma$ and 
$X_{\tau(\omega)}(\omega) \in \partial B'$ by continuity of sample paths. Since 
$\tau < 1-\alpha$ and $\sigma \le 1$, any $\omega \in \Omega_2$ must get from 
$\partial B'$ to $S$. Therefore, as above, $(X_{\tau+s}(\omega))_{s\ge 0} \in 
E$, so that
\begin{equation}
\PP_x(\Omega_2) \le 2^{\frac{d}{2}+1} \exp \left(- \frac{\rho^2}{16\alpha} 
\right).
\end{equation}
Put together, we get the assertion.
\end{proof}

Our main result in this section is

\begin{prop} \label{prop:semidiff}
In the situation above, for $\beta>0$,
\begin{equation} \label{eq:semidiff}
\| e^{-H_{\beta}} - e^{-H_{\Omega,\beta}}\| \le \sqrt{1+ 4\cdot 2^{d/2}} \exp 
\left( -\frac{\rho\sqrt{\beta}}{4\sqrt{2}} \right).
\end{equation}
\end{prop}

\begin{proof}
By the above probabilistic interpretation we get, for $f\in L^2$, $\|f\|_2 \le 
1$ and $x\in \overline{G}$,
\begin{eqnarray}
\lefteqn{\left| e^{-H_{\beta}} f(x) - e^{-H_{\Omega,\beta}} f(x) \right|} 
\notag \\ & = & \left| 
\EE_x \left[ f\circ X_1 \cdot \exp (-\beta T) - f\circ X_1 \cdot \exp (-\beta 
T) \cdot 1_{\{\sigma > 1\}} \right] \right| \notag \\
& = & \left| \EE_x \left[ f\circ X_1 \cdot \exp (-\beta T) \cdot 1_{\{\sigma 
\le 1\}} \right] \right|.
\end{eqnarray}
Therefore, by Cauchy-Schwarz,
\begin{equation} 
\left| (e^{-H_{\beta}} - e^{-H_{\Omega,\beta}}) f(x) \right|^2 \le \EE_x [|f|^2 
\circ X_1] \cdot c(x,\rho,\beta),
\end{equation}
where we have set $c(x,\rho,\beta) := \EE_x \left[ \exp (-2\beta T) \cdot 
1_{\{\sigma \le 1\}} \right]$. Note that $|f|^2 \in L^1$ with $\| |f|^2\|_1 = 
\|f\|_2^2 \le 1$ and that
\begin{equation}
\EE_x \left[ |f|^2 \circ X_1 \right] = e^{-H} (|f|^2)(x).
\end{equation}
Integrating over $G$ gives
\begin{equation}
\| (e^{-H_{\beta}} - e^{-H-{\Omega,\beta}}) f\|_2 \le 1 \cdot \sqrt{\sup_x 
c(x,\rho,\beta)} 
\end{equation}
since $\|e^{-H}:L^1 \to L^1\| \le 1$ as $H$ generates a Dirichlet form. 

We are left with estimating $c(x,\rho,\beta)$ appropriately. To this end we fix 
$\alpha \in (0,1)$, to be specified later, and write
\begin{eqnarray} \label{cest}
c(x,\rho,\beta) & = & \EE_x [\ldots \,|\, T\ge \alpha] + \EE_x[ \ldots \,|\, 
T<\alpha] \notag \\
& \le & \exp (-2\beta \alpha) + \PP_x [ \sigma\le 1, T\le \alpha].
\end{eqnarray}
The second term was estimated in the hit-and-run Lemma by
\begin{equation}
\PP_x [ \sigma \le 1, T \le \alpha] \le 4 \cdot 2^{d/2} \cdot \exp \left( - 
\frac{\rho^2}{16\alpha} \right).
\end{equation}
To get the desired bound on $c(x,\rho,\beta)$ we pick $\alpha$ so as to equate 
exponents in \eqref{cest} above, i.e.,
\begin{equation}
\frac{\rho^2}{16\alpha} = 2\beta \alpha \quad \Longrightarrow \quad \alpha = 
\frac{\rho}{4\sqrt{2}\sqrt{\beta}}.
\end{equation}
Plugged back into \eqref{cest} this gives
\begin{equation}
c(x,\rho,\beta) = (1+4\cdot 2^{d/2}) \exp 
\left(-\frac{\rho\sqrt{\beta}}{2\sqrt{2}} 
\right),
\end{equation}
as was to be shown.
\end{proof}

\section{The Uncertainty Principle: Proof of Theorem~\ref{main}} 
\label{sec:uncert}

In this section we will combine Theorem~\ref{mainsec2} and 
Proposition~\ref{prop:semidiff} with a spectral theoretic uncertainly principle 
from \cite{BLS-10} to derive our main result, Theorem~\ref{main}, a 
quantitative unique 
continuation bound for low energy states of Neumann Laplacians on arbitrary 
convex, not necessarily bounded, subsets $G$ of $\RR^d$. Actually, we will deduce
a slightly stronger, more abstract version in Theorem \ref{main:tec} below, that relates
directly with the spectral uncertainty principle we recall next.

Theorem~1.1 from \cite{BLS-10} refers to a bounded non-negative perturbation 
$W$ of a semibounded self-adjoint operator $H$ in any Hilbert space. If $I$ is 
an interval and $P_I= P_I(H)$ the corresponding spectral projection of $H$, 
then 
it says that
\begin{equation} \label{uncertg}
P_I W P_I \ge \kappa P_I
\end{equation}
as long as there is a $\beta>0$ with
\begin{equation} \label{lambdat}
\mbox{max}\,I < \mbox{min}\, \sigma(H+\beta W) =: \lambda_\beta.
\end{equation}
A lower bound for $\kappa$ is given by
\begin{equation} \label{kappa}
\kappa \ge \sup_{\beta>0} \frac{\lambda_\beta - \mbox{max}\,I}{\beta},
\end{equation}
meaning, in fact, that \eqref{uncertg} holds with $\kappa$ replaced by 
$(\lambda_\beta - \mbox{max}\,I)/\beta$ for every $\beta>0$ which satisfies 
\eqref{lambdat}.

In our application, $H=H^G$ will be the Neumann Laplacian, characterized by the 
quadratic form \eqref{Dirform}, on an open and convex domain $G$ in $\RR^d$. We 
choose $W = 1_B$ as an indicator function of a set $B$ which arises as a 
``fattened'' relatively dense subset of $G$. 

To determine the maximal energy interval $I$ of applicability of 
\eqref{uncertg}, \eqref{lambdat} and \eqref{kappa} in this case, we will need 
to find (at least a lower bound) for 
\begin{equation} \lim_{\beta\to\infty} \lambda_\beta = \lim_{\beta\to\infty} 
\min \sigma(H_\beta),
\end{equation} 
with $H_\beta := H^G + \beta 1_B$.
This will be done in two steps, using our results from 
Sections~\ref{LowerBounds} and \ref{sec:hitrun}: Theorem~\ref{mainsec2} will 
provide a lower bound on the lowest eigenvalue of a mixed Neumann-Dirchlet 
Laplacian, with Neumann condition on $\partial G$ and Dirichlet condition on a 
``semi-fat'' subset $S$ of $B$. Then the norm bound on the difference of 
semi-groups found in Proposition~\ref{prop:semidiff} will yield that this 
eigenvalue is sufficiently close to $\lambda_\beta$, giving the desired lower 
bound for the latter.

In the proof of Theorem~\ref{main:tec} we will use frequently that the first 
Dirichlet 
eigenvalue $\lambda^R$ of
a ball of radius $R$ in $\RR^d$ is given by
\begin{equation}\label{bessel}
j_d R^{-2},
\end{equation}
where $j_d$ is the first positive zero of the Bessel function 
$J_{\frac{d}{2}-1}$. We refrain from telling the whole history 
and refer to the survey article \cite{AB} instead. \\[.5cm]

\begin{theorem} \label{main:tec}
Let $d\ge 3$. Then there exist constants $a,b,C,c>0$, only depending on $d$, 
such that for every open and 
convex $G\subset \RR^d$, any $(R,\delta)$--relatively dense $B$ in $G$, and 
$
\lambda_\beta:= \min \sigma (H^G+\beta 1_B)
$
as above:
\begin{equation} \label{eq:maintec}
\sup_{\beta>0} \frac{\lambda_\beta - E}{\beta} \ge \kappa(R,\delta) ,
\end{equation}
where
\begin{equation} \label{kappa:tec}
E=C\frac{\delta^{d-2}}{R^d} \mbox{  and  
}\kappa(R,\delta)=c\left(\frac{\delta}{R}\right)^d\left[ \frac{b}{(R\wedge 
R_G)^2}
+\left|\log \frac{a\delta^{d-2}}{R^d} \right|\right]^{-2}.
\end{equation}
\end{theorem}
\begin{proof}
To get started with the proof of Theorem~\ref{main:tec}, first note that by 
monotonicity we can replace $B$ by any subset. Thus, without restriction, we 
modify the set--up slightly,
choosing a skeleton $\Sigma\subset B$ for $B$, see Proposition \ref{skeleton}. 
We replace $B$ by $B_\delta(\Sigma)$
and keep the name so that $B$ is now $(3R,\delta)$--dense. Moreover, we set 
$\rho:=\frac{1}{2}\delta$ and $S:=B_\rho(\Sigma)$, so that $S$ is 
$(3R,\rho)$-dense (a ``semi-fat'' subset of $B$).

We may assume further that $\Omega=G\setminus B\not=\emptyset$, as our result 
would be trivial otherwise, giving that 
\begin{equation} \lambda_\Omega:= \inf\sigma(H^{G,B})<\infty.
\end{equation}

From here we proceed in two steps. First, we will prove the Theorem with an 
expression for $\kappa$ where the term $b/(R\wedge R_G)^2$ in \eqref{Ikappa} 
will be replaced by $\lambda_\Omega$. Then we will use some additional 
geometric considerations to get the more explicit final form of \eqref{Ikappa}.

\vspace{.3cm}

\noindent \textbf{1st step:}

\vspace{.3cm}

By estimate (\ref{genlowbound}) from Theorem \ref{mainsec2} we get that 
$\mu_0:= \lags$ satisfies
\begin{equation}\label{eqn1}
\mu_0 \ge c \frac{\delta^{d-2}}{R^d},
\end{equation}
where we have set $c:=d(d-2)/18^d$ (which is not the final value of $c$ in the 
Theorem). 
Our aim is a lower bound for
\begin{equation}
\lambda_\beta= \min \sigma(H^G+\beta 1_B),
\end{equation}
which we achieve by comparing it to
\begin{equation}
\mu_\beta:=\inf\sigma(\hags+\beta 1_{B \setminus S})\ge \mu_0,
\end{equation}
noting that $\lambda_\beta \le \mu_\beta \le \lambda_\Omega$. In fact, the 
difference of the corresponding semigroups is estimated in norm by
\begin{equation}\label{eqn2}
 \| e^{-(H^G+\beta 1_B)}- e^{-(\hags+\beta 1_{B \setminus S})}\|\le 
(1+2^{\frac{d}{2}+2})^\frac12
 \exp 
\left( -\frac{\rho\sqrt{\beta}}{4\sqrt{2}} \right)
\end{equation}
by Proposition \ref{prop:semidiff} from the preceding section. Finally, we pick 
$t\in(0,1)$
and $E_0:=t \mu_0 < \mu_0$, 
so that, by monotonicity,
\begin{equation}
 \label{eqn3}
\mu_\beta-E_0\ge (1-t)\mu_0 .
\end{equation}
If
\begin{equation}
 \label{eqn4}
 \mu_\beta-\lambda_\beta\le \frac12 (1-t)\mu_0,
\end{equation}
we get that 
\begin{equation}
 \label{eqn5}
 \lambda_\beta-E_0\ge \frac12 (1-t)\mu_0>0,
\end{equation}
giving a desired lower bound \ref{kappa:tec} with $\kappa(R,\delta)$ 
determined by the corresponding $\beta$.

Towards \eqref{eqn4}, we observe that (\ref{eqn2}) gives
\begin{equation}
\label{eqn3a}
e^{-\lambda_\beta}-e^{-\mu_\beta}\le A\exp 
\left( -a\rho\sqrt{\beta}\right)
 \end{equation}
 with the obvious (not final) choice of the explicit constants $a,A$.
The mean value Theorem implies that there is $\xi\in [\lambda_\beta,\mu_\beta]$ 
with
\begin{equation}
 \label{eqn7}
\mu_\beta-\lambda_\beta=e^\xi\left(e^{-\lambda_\beta}-e^{-\mu_\beta}\right) \le 
e^{\lambda_\Omega}\left(e^{-\lambda_\beta}-e^{-\mu_\beta}\right).
\end{equation}
Combining (\ref{eqn3a}) and (\ref{eqn7}) we must determine $\beta$ in such a 
way that
\begin{equation}
 \label{eqn6}
 A\exp 
\left( -a\rho\sqrt{\beta}\right)e^{\lambda_\Omega}\le \frac12 (1-t)\mu_0 .
 \end{equation}
Solving for $\beta$ in the previous formula gives
\begin{equation}
 \label{eqn8}
 \beta_0=(a\rho)^{-2}\left[\lambda_\Omega -\log\left( 
\frac{(1-t)\mu_0}{2A}\right)\right]^2 .
 \end{equation}
We plug this value into the right hand side of (\ref{kappa}), using 
(\ref{eqn5}) and obtain
\begin{equation}
\kappa\ge \frac12(1-t)\mu_0(a\rho)^{2}\left[\lambda_\Omega -\log\left( 
\frac{(1-t)\mu_0}{2A}\right)\right]^{-2},
\end{equation}
which gives, by \eqref{eqn1},
\begin{eqnarray}
 \label{eqn9}
 \kappa &\ge& (1-t)\frac12(a\rho)^{2}c 
\frac{\delta^{d-2}}{R^d}\left[\lambda_\Omega 
 -\log\left( \frac{(1-t)c}{2A} \frac{\delta^{d-2}}{R^d}\right)\right]^{-2} 
\nonumber \\
 &=& (1-t) c'\frac{\delta^{d}}{R^d}\left[\lambda_\Omega 
 -\log\left( (1-t)a'\frac{\delta^{d-2}}{R^d}\right)\right]^{-2} ,
 \end{eqnarray}
with constants only depending on $d$,
\begin{eqnarray}
 \label{eqn10}
c'&=& \frac18 a^2c= 2^{-8}\frac{d(d-2)}{18^d}\\
 a'&=&\frac12 c A^{-1}= \frac{d(d-2)}{2\cdot18^d A}\\
 A&=&(1+2^{\frac{d}{2}+2})^\frac12
 \end{eqnarray}
We thus get an uncertainty estimate with
\begin{equation}\label{eqn11}
\sup_{\beta>0} \frac{\lambda_\beta - E_t}{\beta} \ge \kappa_t:= (1-t) 
c'\frac{\delta^{d}}{R^d}\left[\lambda_\Omega 
 -\log\left( (1-t)a'\frac{\delta^{d-2}}{R^d}\right)\right]^{-2}
\end{equation}
valid in the energy range up to 
\begin{equation}\label{eqn12}
 E_t:= t c\frac{\delta^{d-2}}{R^d} .
\end{equation}

On one hand, this is more general than the one we asserted (which we get for 
$t=\frac12$), but not
yet the bound we strive for: the dependence of $\kappa_t$ on $\lambda_\Omega$ 
might
be unpleasant if $\Omega$ is small. On the other hand, this would imply that 
$B$ is large, a situation
which clearly is in favor of our overall result and provides the reason behind 
the following modifications.

\vspace{.3cm}

\noindent \textbf{2nd step:}

\vspace{.3cm}

We now modify $B$ (and $\Omega$) so as to get an upper bound on 
$\lambda_\Omega$. This will require 
some geometrical considerations, partly based on Proposition \ref{skeleton} 
above.

Fix $R_0$ so that
\begin{equation}
\begin{cases}
 \frac14 R_G\le R_0< \frac12 R_G, &\mbox{if  } R_G<\infty\\
 4R\le R_0, &\mbox{if  } R_G=\infty .
\end{cases}
\end{equation}
By definition of $R_G$, in both cases there is $x_0\in G$ such that
\begin{equation}
 \label{eqn13}
 B_{2R_0}(x_0)\subset G .
\end{equation}
We first consider 

\textbf{Case 1:} $4R\le R_0$, including the possibility that $R_G<\infty$. 

\noindent Clearly, in this case the skeleton $\Sigma$ 
introduced at the beginning of the proof must contain at least two elements.

\textbf{Case 1.1:} $\Sigma\cap B_{R_0}(x_0)=\emptyset .$

\noindent Since $\delta\le R$ by definition (and we have set 
$B=B_\delta(\Sigma)$ as before), it follows that
\begin{equation}
\dist(x_0,B)\ge 4R-\delta\ge 3R,
\end{equation}
therefore the open ball $U_{R}(x_0)$ is contained in $G\setminus B$ and so
\begin{equation}
 \label{eqn14}
 \lambda_\Omega\le j_d R^{-2}
\end{equation}
by identity (\ref{bessel}) above. Plugging this bound into estimate 
(\ref{eqn11}) above, we get the assertion of
the theorem with a suitable $b$ since $R^{-2}\le (R\wedge R_G)^{-2}$.

\textbf{Case 1.2:} $\Sigma\cap B_{R_0}(x_0)\not=\emptyset .$

\noindent Choose $s_0\in \Sigma\cap B_{R_0}(x_0)$ and denote $\Sigma_0:= 
\Sigma\setminus\{ s_0\}$ and $B_0:=B_{\delta/2}(\Sigma_0)$. Note that, by 
Proposition \ref{skeleton}, $\dist(s_0,\Sigma_0)\ge R$ and $B_0$ is $(6R, 
\frac12\delta)$--dense. Carrying out the above calculations with this smaller 
subset of $B$, rather than the set $B_\delta(\Sigma)$ used before, we arrive at 
the estimate  (\ref{eqn11}) with $\lambda_\Omega$ replaced by 
$\lambda_{\Omega_0}$ and suitably modified $d$-dependent constants.

We obtain
\begin{equation}
\dist(s_0,B_0)\ge R-\frac{\delta}{2}\ge \frac12 R,
\end{equation}
so that 
\begin{equation}
U_{\frac12 R}(s_0)\subset G\setminus B_0=:\Omega_0,
\end{equation}
giving $\lambda_{\Omega_0}\le bR^{-2}$ and thus the assertion.

\textbf{Case 2:} $R_0<4R$.

\noindent Consequently, $R_G<\infty$, so that $R_0$ and $R_G$ are comparable by 
the definition of $R_0$ above.

 \textbf{Case 2.1:} $\Sigma\cap B_{R_0}(x_0)=\emptyset$. 
 
\noindent This is treated much like Case 1.1 above. In fact, by definition, 
$\delta\le R_G\le 4R_0$. Replacing 
$B=B_\delta(\Sigma)$ by $B=B_{\delta/8}(\Sigma)$, i.e., $\Omega = G \setminus 
B_{\delta/8}(\Sigma)$, we obtain 
\begin{equation}
U_{\frac14 R_0}(x_0) \subset G\setminus B,
\end{equation}
and therefore the assertion with $\lambda_\Omega \le bR_G^{-2}$ for suitable 
$b$.

 \textbf{Case 2.2:} $\Sigma\cap B_{R_0}(x_0)\not=\emptyset$ and $\Sigma$ 
contains at least two elements.
 
\noindent Then we proceed as in Case 1.2 above, this time getting a bound of 
the form  $bR_G^{-2}$. Since no new ideas are
involved we skip the details.

\textbf{Case 2.3:}  $\Sigma=\{ s_0\}\subset B_{R_0}(x_0)$.

\noindent Again replacing $\delta$ by $\frac18\delta$, we see that 
$B_{2R_0}(x_0)\setminus  B_{R_0}(x_0)$
contains a ball of radius $R_0$ that does not intersect with $B$,  once more 
giving a bound of the form  $bR_G^{-2}$ for the corresponding $\lambda_\Omega$.

This completes the proof of Theorem~\ref{main:tec}.
\end{proof}
Combining the previous estimate with the spectral uncertainty principle, 
Theorem 1.1 from
\cite{BLS-10} as explained above, we immediately get:

\begin{corollary}\label{cor:main}
 Let $d\ge 3$. Then there exist constants $a,b,C,c>0$, only depending on $d$, 
such that for every open and 
convex $G\subset \RR^d$, any $(R,\delta)$--relatively dense $B$ in $G$, and 
every selfadjoint operator
$H^\sharp$ satisfying 
$$H^\sharp\ge \eta_0H^G\mbox{  for some }\eta_0>0,$$
\begin{equation} \label{eq:main:tec}
\| f 1_B\|^2  \ge \eta_0\kappa \| f\|^2
\end{equation}
for all $f$ in the range of $P_I(H^\sharp)$, where
\begin{equation} \label{Ikappa:tec}
I=[0,C\eta_0\frac{\delta^{d-2}}{R^d} ]\mbox{  and  
}\kappa=c\left(\frac{\delta}{R}\right)^d\left[ \frac{b}{(R\wedge R_G)^2}
+\left|\log \frac{a\delta^{d-2}}{R^d} \right|\right]^{-2}.
\end{equation}
\end{corollary}
As a special case we obtain our main Theorem \ref{main} stated in the 
introduction. 
Note that 
\begin{itemize}
 \item[(i)] While lower bounds of the form \eqref{eq:main:tec} and 
\eqref{Ikappa:tec} have important applications 
also for the case of bounded sets $G$ (for example for large cubes, where we 
get volume independent bounds), the result is already new and well illustrated 
in the case $G=\RR^d$ or other sets with infinite inradius. In this case it 
gives the following small-$\delta$ and large-$R$ asymptotics:

For fixed $R=R_0$ we have $\kappa \sim \delta^d/|\log \delta|$ on $I=[0,C 
\delta^{d-2}]$ as $\delta\to 0$. 

For fixed $\delta = \delta_0$ we have $\kappa \sim \frac{1}{R^d (\log R)^2}$ on 
$I =[0,C R^{-d}]$ as $R\to \infty$.

\item[(ii)] In principle, our methods could also be used to get bounds for 
$d=2$, but 
the constants would look less satisfying (and contain more logarithms). 

\item[(iii)] Totally different methods are available for $d=1$; see 
\cite{KirschVeselic-02}.
\end{itemize}

We refrain from spelling out more consequences in form of Corollaries and 
instead list a few more possibilities of exploiting the flexibility of the 
preceding Corollary.

\begin{itemize}
 \item We can regard different b.c., in particular periodic b.c. in case that 
$G$ is a cube
 and obtain the same estimates as above for the corresponding operator 
$H^G_{b.c.}$.
 \item We can add a nonnegative potential $V$ and get the same estimates as 
above for the corresponding operator $H^G_{b.c.}+V$.
 \item More generally, not necessarily positive lower order terms that are controlled by $H^G$ can be added, i.e., we
 can treat $H^G+B$ as long as $B\ge - \gamma H^G$ for some $\gamma<1$. 
\end{itemize}

We end our discussion by mentioning that our above results can be used to
prove \textbf{localization} (see \cite{Kirsch,Stollmannbook} for the general phenomenon of
bound states for random models) for new classes of random models. As remarked in
the introduction, uncertainty principles are used to derive Wegner and Lifshitz tail estimates
when the random perturbation obeys no covering condition. With the uniform estimates above, one could treat
models with a random second order main term plus a random potential.

\begin{appendix}

\section{Capacities of balls in $\RR^d$} \label{capacities}

As compared to the discrete case of graphs, euclidean space is more complicated 
in many ways. One important difference that matters for our analysis is that 
points are not massive at all, at least in dimension $d\ge 2$. This is why a 
finite inradius of an open set $\Omega \subset \RR^d$ does not imply that $\inf 
\sigma(-\Delta_{\Omega}) >0$ for the {\it Dirichlet 
Laplacian} 
$-\Delta_{\Omega}$, defined via forms as the Friedrichs extension of 
$-\Delta$ 
on $C_c^{\infty}(\Omega)$, or, equivalently, as the selfadjoint operator 
associated with the form
\begin{equation}
\en{[u]} := \int_{\Omega} |\nabla u(x)|^2\,dx \; \mbox{on 
$W_0^{1,2}(\Omega)$.}
\end{equation}

\begin{example}
In $\RR^d$ for $d\ge 2$ consider $D=\ZZ^d$ and the union of closed balls $S:= 
\cup_{k\in D} 
B_{r_k}(k)$ with $0<r_k < \frac{1}{2}$ for $k\in D$. For
\begin{equation} 
\Omega := \RR^d \setminus S
\end{equation}
we see that the inradius $R_\Omega = \sup \{s>0\,|\,\exists \,x\in \Omega: 
B_s(x) \subset \Omega\}$ is bounded above by $\sqrt{d}/2$. However, as we will 
see below,
\begin{equation} \label{notmassive} 
\capa{(B_r(x))} = 
\capa{(B_r(0))} \to 0 \quad \mbox{as $r\to 0$},
\end{equation} 
so that we can pick 
$r_k$ such that
\begin{equation}
\capa{(S)} \le \sum_k \capa{(B_{r_k}(k))} < \infty.
\end{equation} 
In that case, by Theorem 1 in \cite{StollmannJFA}, we get that $e^{-\Delta} - 
e^{-\Delta_{\Omega}}$ is Hilbert-Schmidt and therefore 
$\sigma_{ess}(-\Delta_{\Omega}) = \sigma(-\Delta_{\Omega}) = [0,\infty)$.
\end{example}

As different notions of capacity are around, let us briefly settle the case of 
\eqref{notmassive} above: 

In the above result, capacity refers to the $1$--capacity, often used in 
potential theory 
for Dirichlet forms and defined by the following variational principle,

\begin{equation}
 \label{cap}
 \capa(B_r(0)):=\inf \left\{ \|\nabla f\|^2 + \|f\|^2\mid f\in C^1_c(\RR^d), 
f\ge 1_{B_r(0)}\right\} . 
\end{equation}

Set 
$\phi(x)=1$ if $0\le x \le 1$, $\phi(x)=2-x$ if $1\le x \le 2$ and $\phi(x)=0$ 
for $x>2$ and define $f_r(x) = \phi(|x|/r)$ on $\RR^d$. Then $\capa(B_r(0)) \le 
\|\nabla f_r\|^2 + \|f_r\|^2 \le C_d r^{d-2}$, which gives the claim for $d\ge 
3$. In $d=2$ this only gives boundedness, but can be combined with $\|f_r\|^2 
\to 0$, weak compactness of the unit ball in $W^{1,2}$ and Hahn-Banach to give 
a sequence $r_n$ with $\capa{(B_{r_n}(0))}\to 0$, proving \eqref{notmassive} by 
monotonicity of the capacity. 

We go on to show that for $d\ge 3$,
\begin{equation}
 \label{capofballs}
 \capa(B_r(0))\sim r^{d-2} \mbox{  for  }r\le 1 .
\end{equation}
This is most easily seen by using the slightly smaller Newtonian capacity

\begin{equation}
 \label{newcap}
 \capa_N(B_r(0)):=\inf \left\{ \|\nabla f\|^2 \mid f\in C^1_c(\RR^d), 
f\ge 1_{B_r(0)}\right\} . 
\end{equation}
The above scaling shows immediately, that $\capa_N(B_r(0))\sim r^{d-2}$, so that
\eqref{capofballs} follows, since $\capa_N(B_r(0))\le \capa(B_r(0))$. We cannot resist
to mention two classical papers on capacities, \cite{Riesz,PS}. For a thorough discussion,
we refer to Section 11.15 in \cite{LL}, as well as to classical textbooks like \cite{Landkov}.

\end{appendix}

\bigskip
\footnotesize
\noindent\textit{Acknowledgments.}
Many thanks go to Marcel Reif for most 
valuable comments over the
years and Wolfgang Löhr for an inspiring discussion concerning reflected 
Brownian motion and the
strong Markov property.

\end{document}